\documentclass[aps,pra,superscriptaddress,twocolumn]{revtex4-1}
\usepackage{amsmath,amssymb,amsthm,easybmat,verbatim}
\usepackage{enumerate}
\usepackage{color}
\usepackage{bm}
\usepackage{graphicx}
\usepackage{longtable}
\usepackage{array}
\newcolumntype{x}[1]{%
>{\centering\arraybackslash}p{#1}}%

\newtheorem{lemma}{Lemma}
\newtheorem{proposition}{Proposition}
\newtheorem{theorem}{Theorem}

\theoremstyle{definition}
\newtheorem{definition}{Definition}
\newcommand{\bra}[1]{\langle #1|}
\newcommand{\ket}[1]{|#1\rangle}

\newcommand{\ip}[2]{\langle #1|#2\rangle}
\newcommand{\op}[2]{|#1\rangle \langle #2|}

\newcommand{\mc}[1]{\mathcal{#1}}
\newcommand{\mbf}[1]{\mathbf{#1}}
\newcommand{\mbb}[1]{\mathbb{#1}}
\newcommand{\supp}{\text{\upshape supp}}

\newcommand{\LU}{\mathbf{L}\mathbf{U}}
\newcommand{\Swap}{\mathbf{S}\mathbf{W}\mathbf{A}\mathbf{P}}

\newcommand{\CUU}{\mathbf{U}_{\mbf{C}2}}
\newcommand{\LUeq}{\overset{\underset{\mathrm{LU}}{}}{\approx}}

\begin{document}

\title{Accurate Modeling of Reduced-State Dynamics}

\author{Eric Chitambar}
\affiliation{Department of Physics and Astronomy{\mbox ,} Southern Illinois University, 
Carbondale, Illinois 62901, USA}
\email{echitamb@siu.edu}
\author{Ali Abu-Nada}
\affiliation{Department of Physics and Astronomy{\mbox ,} Southern Illinois University, 
Carbondale, Illinois 62901, USA}
\affiliation{Department of Physics and Engineering Physics{\mbox ,} Southeast Missouri State University, Cape Girardeau, MO 63701, USA}   
\author{Russell Ceballos}
\affiliation{Department of Physics and Astronomy{\mbox ,} Southern Illinois University, 
Carbondale, Illinois 62901, USA}
\author{Mark Byrd}
\affiliation{Department of Physics and Astronomy{\mbox ,} Southern Illinois University, 
Carbondale, Illinois 62901, USA}

\date{\today}

\begin{abstract}
In this paper we return to the problem of reduced-state dynamics in the presence of an interacting environment.  The question we investigate is how to appropriately model a particular system evolution given some knowledge of the system-environment interaction.  When the experimenter takes into account certain known features of the interaction such as its invariant subspaces or its non-local content, it may not be possible to consistently model the system evolution over a certain time interval using a standard Stinespring dilation, which assumes the system and environment to be initially uncorrelated.  Simple examples demonstrating how restrictions can emerge are presented below.  When the system and environment are qubits, we completely characterize the set of unitaries that always generate reduced dynamics capable of being modeled using a consistent Stinespring dilation.  Finally, we show how any initial correlations between the system and environment can be certified by observing the system transformation alone during certain joint evolutions.
\end{abstract}

\maketitle

\section{Introduction}

In quantum mechanics, the time evolution of a closed system is described by a unitary transformation acting on its state space.  However, when the system is interacting with some external environment, its reduced-state dynamics is no longer unitary.  How to properly characterize the dynamics of such open quantum systems has been an area of extensive research and a source of lively debate \cite{Sudarshan-1961a, Pechukas-1994a, Alicki-1995a, Stelmachovic-2001a, Jordan-2004a, Carteret-2008a, Rodriguez-Rosario-2010a, Modi-2012a, Dominy-2014a}.  In general, an overall unitary evolution of the combined system-environment translates into a reduced dynamics of the system given by
\begin{equation}
\label{Eq:ReducedDynamics1}
\rho_{S}(t_0)\to\rho_S'(t)=tr_E[U_{SE}(\rho_{SE}(t_0))U_{SE}^\dagger].
\end{equation}
All is well if the system and environment are known to begin interacting at some time $t_0$, prior to which they are uncorrelated; i.e. $\rho_{SE}(t_0)=\rho_S\otimes \sigma_E$.  Then Eq. \eqref{Eq:ReducedDynamics1} has the standard Stinespring form, and thus the system evolution can be described by a completely positive (CP) map $\mc{E}(\rho_S)=tr_E[U_{SE}(\rho_S\otimes\sigma_E)U_{SE}^\dagger]$ \cite{Stinespring-1955a, Nielsen-2000a}.  But difficulty arises when the system and environment are initially correlated, and one must be careful when trying to interpret Eq. \eqref{Eq:ReducedDynamics1} as a map acting on the system's state space \cite{Dominy-2014a}.  %Types of initially correlated states $\rho_{SE}$ having particular physical significance are the classically-correlated states, the separable or unentangled states, and the entangled states.

In this paper, we are not interested in understanding Eq. \eqref{Eq:ReducedDynamics1} as anything more than a physical model for the particular state transformation $\rho_S\to\rho_S'$.  Although we will make no specific reference to ``maps'' in this scenario, the transformation problem fits within more general frameworks designed to handle restricted-domain subsystem maps \cite{Pechukas-1994a, Alicki-1995a, Modi-2012a, Dominy-2014a, Buscemi-2014a, Dominy-2015a}.  We suppose that an experimenter measures the system to be in state $\rho_S$ at time $t_0$ and in state $\rho'_S$ at some later time $t$; Eq. \eqref{Eq:ReducedDynamics1} then offers a physical description of how this transformation came about in terms of an initial correlation with the environment ($\rho_{SE}$) and a subsequent interaction ($U_{SE}$).  There, of course, will be many different choices of $\rho_{SE}$ and $U_{SE}$ that will successfully model the observed transformation $\rho_S\to\rho_S'$ via Eq. \eqref{Eq:ReducedDynamics1}.  However, if the experimenter has some knowledge of $U_{SE}$ it will greatly limit the possible $\rho_{SE}$.  

The primary goal of this paper is to begin understanding what type of initial system-environment states $\rho_{SE}$ will correctly model a given reduced state transformation $\rho_S\to\rho_S'$ if the joint unitary $U_{SE}$ is \textit{a priori} known or at least partially known.  For example, if $U_{SE}$ is a $d\otimes d$ unitary with known eigenstates $\ket{\varphi_i}$ and $\rho_S\to\rho_S'$ is an observed reduced-state transformation under $U_{SE}$, what are the compatible initial states $\rho_{SE}$?   We study this question and provide a general necessary condition relating the eigenstates of $U_{SE}$ with any $\rho_{SE}$ that generates the reduced-state transformation $\rho_S\to\rho_S'$.

We also consider permissible models for $\rho_S\to\rho_S'$ if a certain nonlocal character is demanded of $U_{SE}$.  Joint unitaries acting on $\mc{H}_S\otimes\mc{H}_E$ can be grouped into  local unitary (LU) equivalence classes such that $U_{SE}\LUeq U_{SE}'$ iff there exists product unitaries $V=V_S\otimes V_E$ and $W=W_S\otimes W_E$ such that $U_{SE}=WU_{SE}'V$.  Two unitaries that are not LU equivalent can be regarded as possessing different types of nonlocality since the action of one cannot be simulated using the other combined with arbitrary local unitaries.  If it is known that the system-environment unitary belongs to a certain equivalence class, is it always possible to choose a product state $\rho_{SE}$ to model an observed transformation $\rho_S\to\rho_S'$?  Below we show that when $\mc{H}_S\otimes\mc{H}_E$ is two qubits, this can be done iff $U_{SE}$ is LU equivalent to either the swap operator or the identity.  Even stronger, for every other type of unitary, there exists system transformations that require an initially entangled $\rho_{SE}$ to accurately model.  

One particularly interesting transformation we consider involves converting a mixed system state into a pure one via Eq. \eqref{Eq:ReducedDynamics1}, a process we generically call \textit{purity extraction}.  If we assume that $\rho_{SE}=\rho_S\otimes\sigma_E$ is a product state, what LU classes of unitaries will generate the reduced-state purity extraction $\rho_S\to\op{\psi}{\psi}$?  We identify the class of unitaries in two-qubits that can be used to model such a process.  Studying these types of transformations may have thermodynamics applications as purity obtained from a quantum system can then be used to perform work \cite{Horodecki-2003b, Horodecki-2005a, delRio-2011a}.

As an application of this line of inquiry, we consider using the knowledge of $U_{SE}$ to certify initial correlations between the system and environment through measuring exclusively the system.  For example, suppose that in Eq. \eqref{Eq:ReducedDynamics1} we take $U_{SE}$ to be the two-qubit CNOT gate and $\rho_{SE}$ to be the entangled pure state $\ket{\psi}=\sqrt{1/2}(\ket{0+}+\ket{1-})$, where $\ket{\pm}=\sqrt{1/2}(\ket{0}\pm\ket{1})$.  Then the induced transformation on $S$ is $\mbb{I}/2\to\op{0}{0}$.  It is not difficult to show that this transformation is impossible for any product state input $\rho_{SE}= \mbb{I}/2\otimes\sigma_E$ whenever CNOT is being implemented.  This means that if we know the joint dynamics to be governed by $U_{SE}$, but we do not know the initial joint state $\rho_{SE}$, then detecting the system transformation $\mbb{I}/2\to\op{0}{0}$ means that $S$ and $E$ cannot be uncorrelated in the initial state.  In fact, from the analysis of Sect. \ref{Sect:Two-qubit examples} it can further be shown that $\rho_{SE}$ must be entangled to witness the transformation $\mbb{I}/2\to\op{0}{0}$ under CNOT.  The general question then becomes the following: Given a state $\rho_S$, can one always find a unitary $U_{SE}$ such that a transformation \textit{\`{a} la} Eq. \eqref{Eq:ReducedDynamics1} is possible \textit{only if $\rho_{SE}$ is entangled}?  If so, then detecting the system transformation $\rho_S\to\rho'_S$ would necessarily indicate that the system is initially entangled with the environment.  In this paper, we construct a unitary $U_{SE}$ for every $\rho_S$ that generates a transformation $\rho_S\to\rho_S'$ which is possible only if $\rho_S$ is initially entangled with the environment.  By the same reasoning, we also construct a very general experimental procedure for certifying when $\rho_{SE}$ fails to be a product state, based on measuring the system state alone. 

Before proceeding in more detail, we introduce the following definition which fixes the language used to describe Eq. \eqref{Eq:ReducedDynamics1}.
\begin{definition}
Let $U$ be any unitary acting on $\mc{H}_S\otimes\mc{H}_E$.  We say that a transformation $\rho_S\to\rho'_S$ can be \textbf{$U$-generated by $\rho_{SE}$} if (i) $tr_E[\rho_{SE}]=\rho_S$ and (ii) $tr_E[U\rho_{SE}U^\dagger]=\rho_S'$.  The transformation $\rho_S\to\rho'_S$ is a \textbf{$U$-generated physical transformation} if it is $U$-generated by some density matrix $\rho_{SE}$.
\end{definition}
We now begin in Sect. \ref{Sect:Symm} by considering the restrictions in modeling a physical transformation when invariances of $U_{SE}$ are known.  We will then consider in Sect. \ref{Sect:LU} modeling transformations using unitaries from certain LU equivalence classes.  Finally, in Sect. \ref{Sect:EntanglementDetection} we turn to the question of detecting entanglement by observing reduced-state dynamics.

\section{Modeling with Known Eigenstates of $U_{SE}$: General Restrictions on the Spectrum or $\rho_S'$}
\label{Sect:Symm}

In this section we prove a general necessary condition for the permissible joint states that generate a particular transformation when eigenstates of the unitary are known.
\begin{lemma}
\label{Lem:Ueigenstate}
Suppose that $\ket{\varphi}_{SE}$ is an eigenstate of $U$; i.e. $U\ket{\varphi}=e^{i\theta}\ket{\varphi}$.  Then a transformation $\rho_S\to\rho_S'$ can be $U$-generated by $\rho_{SE}$ only if 
\begin{equation}
\tau\cdot\lambda_k(tr_E\op{\varphi}{\varphi})\leq \lambda_k(\rho_S')\notag
\end{equation} 
for all $k\in\{1,\cdots,rk[\rho_{S}]\}$, $\lambda_k(\sigma)$ denotes the $k^{th}$ largest eigenvalue of $\sigma$, and 
\[\tau=\begin{cases}\tfrac{1}{\bra{\varphi}\rho_{SE}^{-1}\ket{\varphi}}\quad\text{if $\ket{\varphi}\in supp(\rho_{SE})$}\\0 \quad\text{otherwise}.\end{cases}
\]
\end{lemma}
\begin{proof}
Suppose that $\rho_S\to\rho_S'$ is $U$-generated by $\rho_{SE}$.  If $\ket{\varphi}\not\in\supp(\rho_{SE}$), then the lemma trivially holds.  So assume $\ket{\varphi}\in\supp(\rho_{SE}$).  Taking a spectral decomposition $\rho_{SE}=\sum_i p_i\op{e_i}{e_i}$, we can write $\ket{\varphi}=\sum_i\alpha_i\ket{e_i}$ for some coefficients $\alpha_i$.  By the Schr\"{o}dinger-HJW Theorem, it follows that we can expand $\rho_{SE}=q \op{\varphi}{\varphi}+(1-q)\sigma$ for some density operator $\sigma$ whenever $\sqrt{q}\ket{\varphi}=\sum_{i}u_i\sqrt{p_i}\ket{e_i}$ with $\max_i|u_i|^2= 1$ \cite{Hughston-1993a}.  Since $\sqrt{q}\ket{\varphi}=\sum_i\left(\sqrt{\frac{q}{p_i}}\alpha_i\right)\sqrt{p_i}\ket{e_i}$, we have that $q=\min_i\tfrac{p_i}{|\alpha_i|^2}\geq \tfrac{1}{\bra{\varphi}\rho_{SE}^{-1}\ket{\varphi}}$.  Since the transformation is generated by $\rho_{SE}$, we have
\begin{align}
tr_E[U\rho_{SE} U^\dagger]=q\cdot tr_E\op{\varphi}{\varphi}+(1-q)tr_E[U\sigma U^\dagger]=\rho_S'.\notag
\end{align}
The proposition then follows from Weyl's Theorem, which gives that $\lambda_k (A)\leq \lambda_k(A+B)$ for non-negative operators $A$ and $B$ \cite{Horn-1985a}. 
\end{proof}
\noindent Note the bound of Lemma \ref{Lem:Ueigenstate} can be made trivially tight by taking $\rho_{SE}=\op{\varphi}{\varphi}$.

\subsubsection{Examples}

To demonstrate the use of Lemma \ref{Lem:Ueigenstate}, consider any two-qubit unitary with a maximally entangled eigenstate $\ket{\varphi}=\sqrt{1/2}(\ket{00}+\ket{11})$.  Suppose the initial state has the form $\rho_{SE}=\mbb{I}/2\otimes(\mbb{I}/2+\vec{m}\cdot\vec{\sigma})$, and the interaction is observed to generate the final system state $\rho_{S}'=p\op{0}{0}+(1-p)\op{1}{1}$.  Are there any constraints on the possible form of $\rho_E$?  Since $\tau=1/tr(\rho_E^{-1})=1/4-m_1^2-m_2^2-m_3^2$, we can deduce from Lemma \ref{Lem:Ueigenstate} the necessary condition that 
\[7/8+1/2\sum_{i=1}^3m_i^2\geq p\geq 1/8-1/2\sum_{i=1}^3m_i^2.\]
On the other hand, suppose that $\rho_{SE}=\mbb{I}/2\otimes\mbb{I}/2+\sum_{i=1}^3m_i\sigma_i\otimes\sigma_i$, so that $\rho_S$ is still the maximally mixed state.  Then $\tau=1/4+m_1-m_2+m_3$, and 
\[7/8-1/2(m_1-m_2+m_3)\geq p\geq 1/8+1/2(m_1-m_2+m_3).\]
As $(m_1,m_2,m_3)\to(1/4,-1/4,1/4)$, the only compatible values of $p$ converge to $1/2$, which is expected since $\rho_{SE}\to \op{\varphi}{\varphi}$.

\subsubsection{Application: Increasing Purity Using an Interacting Bath}

The previous example can be interpreted as providing bounds on how close the state $\rho_S'$ can be brought to a pure state, given the fact that $U_{SE}$ acts invariantly on a maximally entangled state.  We generalize this idea here.

For a general system state $\rho_S$, any function depending on its eigenvalues remain invariant under unitary evolution.  Examples of such functions are the Von Neumann entropy and its first-order approximation approximation, the linear entropy.  The linear entropy of $\rho_S$ is given by $1-\gamma(\rho_S)$, where $\gamma(\rho_S)$ is the so-called purity of $\rho_S$:
\[\gamma(\rho_S)=tr[\rho_S^2]=\sum_{k=1}^{rk[\rho_S]}[\lambda_k(\rho_S)]^2.\]
Due to unitary invariance of $\gamma$, the purity of a given state can only be increased through an interacting external environment via Eq. \eqref{Eq:ReducedDynamics1}.  How large can the purity be increased given known conditions of $\rho_{SE}$ and $U$?

We can use Lemma \ref{Lem:Ueigenstate} to compute a general upper bound on the increase in purity of $\rho_{S}'$ if the initial system-environment state is assumed to take the product state form $\rho_{SE}=\rho_S\otimes\mbb{I}/d_E$.  First assume that $\rho_S$ is invertible.  For an arbitrary $\ket{\varphi}\in\mc{H}_S\otimes\mc{H}_E$, we have
\begin{align}
\bra{\varphi}\rho_{SE}^{-1}\ket{\varphi}&=d_E tr[\rho_S^{-1} tr_E(\op{\varphi}{\varphi})]\notag\\
&\leq d_E tr[\rho_S^{-1}]\leq d_Ed_S/\lambda_r(\rho_S),\notag
\end{align}
where $r=rk[\rho_S]$.  Hence for every eigenstate $\ket{\varphi}$ of $U_{SE}$, Lemma \ref{Lem:Ueigenstate} gives
\begin{equation}
\lambda_k(\rho_S')\geq\frac{\lambda_r(\rho_S)}{d_Ed_S}\lambda_k(tr_E\op{\varphi}{\varphi}).
\end{equation}
In terms of the purity, we therefore obtain 
\begin{equation}
\label{Eq:PurityBound}
\gamma(\rho_S')\leq 1-\left(\frac{\lambda_r(\rho_S)}{d_Ed_S}\right)^2\max_{\varphi}\sum_{k=1}^{rk[\rho_S]}\lambda_k(tr_E\op{\varphi}{\varphi})^2,
\end{equation}
where the maximization is taken over all eigenstates of $U$.  Note that this upper bound trivially holds when $\rho_S$ is not invertible.  Eq. \eqref{Eq:PurityBound} then gives a general limit to how much the purity of $\rho_S$ can be increased by interaction with an initially uncorrelated bath, $\mbb{I}/d_E$, when one or more eigenstates of the interacting unitary are known.

\section{Modeling with Different Classes of Unitaries: Two-Qubit Solutions}

\label{Sect:LU}

We now consider the various restrictions that emerge for the structure of $\rho_{SE}$ when modeling a given system transformation with unitaries belonging to different LU equivalence classes.  In particular, we are interested in the following question: For some unitary $U_{SE}$, can every $U$-generated physical transformation be generated by a product state?  If not, what are the types of unitaries for which this is possible?

\subsection{General Conditions for Simulating without Initial Correlations}

When the system and environment form a two-qubit system, the above two questions can be completely solved.  To state the result, we begin by recalling two special types of two-qubit unitaries: (1)  The SWAP operator $\mbb{F}$ is the unitary map whose action is $\mbb{F}\ket{\alpha\beta}=\ket{\beta\alpha}$ for any product state $\ket{\alpha\beta}$, (2) An $SE$-controlled unitary operator $U_c^{SE}$ is any unitary of the form $U^{SE}_c=\op{0}{0}_S\otimes \mbb{I}+\op{1}{1}_S\otimes V_E$ where $U_E$ and $V_E$ are unitaries acting on the environment.  Note that every $U^{SE}_c$ is defined with respect to some fixed computational basis and with $S$ being the control and $E$ the target.  When the roles of control and target are reversed, we have an $ES$-controlled unitary $U_c^{ES}$, which therefore has the form $U^{ES}_c=\mbb{F}U^{SE}_c\mbb{F}$ for some $SE$-controlled unitary $U^{SE}_c$.  These unitaries can be used to represent special equivalence classes of two-qubit unitaries.  Letting $\LU$ denote the set of two-qubit local unitaries, we define the following:
\begin{align}
\Swap&:=\{U: \text{$U\LUeq\mbb{F}$}\},\notag\\
%\CU&:=\{U: \text{$U\LUeq U_c^{SE}$ for any $U_c^{SE}$}\},\notag\\
\CUU&:=\{U: \text{$U\LUeq U_c^{SE}U_c^{ES}$ for any $U_c^{SE}$ and $U_c^{ES}$}\}.\notag
\end{align}
  It can  be shown that $\Swap\not=\CUU$.  Finally, turning to the initial system-environment state, the state $\rho_{SE}$ is separable (or unentangled) if it can be expressed as a convex combination of pure product states:
\[\rho_{SE}=\sum_kp_k\op{\alpha_k}{\alpha_k}_S\otimes\op{\beta_k}{\beta_k}_E.\]
If the state can be decomposed in such a way that the $\ket{\beta_k}$ are also pairwise orthogonal, then the states is said to be quantum-classical (QC).  With these classifications, we can now state the main result of this section.
\begin{theorem}
\label{Thm:Main}
Suppose that the system and environment consists of two qubits.  Every $U$-generated physical transformation $\rho_S\to\rho_S'$ can be $U$-generated by a product state iff $U$ belongs to $\LU\cup\Swap$.  If $U$ belongs to $\CUU$, the transformation can be $U$-generated by a QC state.  On the other hand, if $U$ does not belong to $\Swap\cup\CUU$, then there exists physical transformations that cannot be $U$-generated by any separable state.
\end{theorem}

The full proof is carried out in the Appendix.  The analysis there relies heavily on the special structure of two-qubit systems.  For instance, there exists a so-called ``magic basis'' of $2\otimes 2$ systems given by
\begin{align}
\ket{\Phi_1}&=\sqrt{1/2}(\ket{01}-\ket{10})&\ket{\Phi_2}&=-i\sqrt{1/2}(\ket{00}-\ket{11})\notag\\
\ket{\Phi_3}&=\sqrt{1/2}(\ket{00}+\ket{11}&\ket{\Phi_4}&=-i\sqrt{1/2}(\ket{01}+\ket{10}).\notag
\end{align}
Working in the magic basis has proven to be very helpful in the study of two-qubit entanglement \cite{Bennett-1996a, Hill-1997a}.  For instance, the concurrence of a pure state $\ket{\Psi}$ is given by $C(\Psi)=\sum_{i=1}^4|c_i^2|$ where the $c_i$ are the component of $\ket{\Psi}$ when expressed in the magic basis \cite{Hill-1997a}.  From this it follows that $\ket{\Psi}$ is a product state iff $\sum_{i=1}^4 c_i^2=0$.  Furthermore, in Ref. \cite{Kraus-2001a}, Kraus and Cirac have shown that every two-qubit unitary $U$ can be decomposed as
\begin{equation}
\label{Eq:UnitaryCanonical0}
U=(U_S\otimes U_E) U_d (V_S\otimes V_E),
\end{equation}
where $U_d$ is diagonal in the magic basis: 
\begin{equation}
\label{Eq:UnitaryCanonical}
U_d=\sum_{i=1}^4 e^{-i\lambda_i}\op{\Phi_i}{\Phi_i},
\end{equation}
and $0\leq\lambda_i<2\pi$.   It should be noted that the matrix $U_d$ is not unique for a given $U$ since the $\ket{\Phi_i}$ can be interconverted by a local unitary.  With the form of a general two-qubit unit greatly simplified by Eq. \eqref{Eq:UnitaryCanonical}, we compute in the Appendix the various types of transformations that can be $U_d$-generated for a given $U_d$.  The transformations that are shown to require an entangled initial state $\rho_{SE}$ involve transforming a rank-two state $\rho_S$ into a pure state $\rho_S'=\op{\psi}{\psi}$.

\subsubsection*{Robustness of Theorem \ref{Thm:Main}}

We now give a simple continuity argument showing that the above result holds even for nonzero error in the initial and final states.  Namely, for every unitary in $U\not\in\Swap\cup \CUU$, there exists a limit to how well every $U$-generated physical transformation can be approximated using an uncorrelated system-environment state.   Suppose that $U\not\in\Swap\cup \CUU$ is given.  By Theorem \ref{Thm:Main}, there exists a transformation $\rho_S\to\rho_S'$ that cannot be $U$-generated by a product state.  Letting $\mc{D}(\mbb{C}^2)$ denote the set of one-qubit density matrices, define the map $\varphi:\mc{D}(\mathbb{C}^2\otimes\mathbb{C}^2)\to \mc{D}(\mathbb{C}^2)$ given by $\varphi\left(\rho_{SE}\right)=tr_E[U_{SE}(\rho_{SE})U^\dagger_{SE}]$, which is uniformly continuous (it is a linear map on a finite, compact set).  Next, let 
\[\epsilon=\min_{\omega_E\in \mc{D}(\mathbb{C}^2)} ||\phi\left(\rho_S\otimes\omega_E\right)-\rho_S'||_1,\]
where the minimum indeed exists due to compactness of $\mc{D}(\mathbb{C}^2)$, and $\epsilon>0$ due to Theorem \ref{Thm:Main}.  Here, we are using $||A||_1=tr\sqrt{A^\dagger A}$ to denote the trace norm.
 Therefore, by uniform continuity of $\varphi$, there exists a $\delta>0$ such that for \textit{any} $\omega_E\in \mc{D}(\mathbb{C}^2)$ we have
\[||\rho_{SE}-\rho_S\otimes\omega_E||_1<\delta\qquad\Rightarrow\qquad ||\phi\left(\rho_{SE}\right)-\rho_S'||_1>\epsilon/2.\]
Hence, as long as the prepared state $\rho_{SE}$ is within $\delta$-distance (w.r.t. the trace norm) to any product state of the form $\rho_S\otimes\omega_E$, we are guaranteed that the final state will be at least $\epsilon/2$-distance away from the target state $\rho_S'$.  Note that $\rho_{SE}$ need not be a product state for this error bound to hold.

\subsection{Purity Extractions in Two Qubits with No Initial Correlations}

The main technique used to prove Theorem \ref{Thm:Main} involves constructing, for a given $U_d$, one particular transformation of a mixed state $\rho_S$ into a pure one that is impossible unless $\rho_{SE}$ has a certain form.  Let us consider the transformation $\rho_S\to\op{\psi}{\psi}$ in more generality, first by focusing on the situation when the system and environment are initially in a product state.
\begin{lemma}
\label{Lem:PurityExtraction}
Suppose that the system and environment consists of two qubits.  If $\rho_S\to\op{\psi}{\psi}$ is $U$-generated by some product state $\rho_S\otimes \rho_E$ with $rk(\rho_S)=2$, then $U\in\Swap$.
\end{lemma}
\begin{proof}
The joint transformation is $U(\rho_S\otimes\rho_E) U^\dagger=\op{0}{0}\otimes\omega_E$, which requires that $\rho_E$ is pure while $\omega_E$ and $\rho_S$ must have the same spectrum.  Hence, the action of $U$ must take the form
\begin{align}
\ket{0}\ket{\eta}&\overset{\underset{U}{}}{\longrightarrow}\ket{\psi}\ket{0'}&\ket{0}\ket{\eta^\perp}&\overset{\underset{U}{}}{\longrightarrow}\ket{\psi^\perp}\ket{0''}\notag\\
\ket{1}\ket{\eta}&\overset{\underset{U}{}}{\longrightarrow}\ket{\psi}\ket{1'}&\ket{1}\ket{\eta^\perp}&\overset{\underset{U}{}}{\longrightarrow}\ket{\psi^\perp}\ket{1''}.
\end{align}
Up to local unitaries, this requires $U$ to be $\mbb{F}$.
\end{proof}

\subsection{A Family of Two-Qubit Purity Extractions}

\label{Sect:Two-qubit examples}

Next, we turn to purity extraction for an arbitrary initial $\rho_{SE}$.  For simplicity, we focus on the specific $U$-generated transformation $\mbb{I}/2\to\op{0}{0}$ with $U_{SE}$ belonging to a two-parameter family of unitaries given by
\begin{align}
U(\theta,\gamma)=\begin{pmatrix}\tfrac{1}{\sqrt{2}}&0&0&\tfrac{1}{\sqrt{2}}\\0&\cos\theta&\sin\theta&0\\0&-e^{i\gamma}\sin\theta & e^{i\gamma}\cos\theta&0\\\tfrac{1}{\sqrt{2}}&0&0&-\tfrac{1}{\sqrt{2}}\end{pmatrix}.
\end{align}

Since the unitary $U$ preserves the rank of $\rho_{SE}$, the transformation $\mbb{I}/2\to\op{0}{0}$ in two qubits is possible only if $\rho_{SE}$ has rank at most two.  Let $\ket{e_1}$ and $\ket{e_2}$ being the eigenstates of $\rho_{SE}$ with corresponding eigenvalues $p_1$ and $p_2$.  Then $\ket{0}$ being the final system state requires that
\begin{align}
U\ket{e_1}&=\ket{0}(\cos\mu\ket{0}+e^{i\nu}\sin\mu\ket{1})\notag\\
U\ket{e_2}&=\ket{0}(-\sin\mu\ket{0}+e^{i\nu}\cos\mu\ket{1})\notag
\end{align}
for some values $\mu$ and $\nu$.  Applying $U^\dagger$ to both sides gives 
\begin{align}
\ket{e_1}&=\cos\mu\ket{\Phi_3}+e^{i\nu}\sin\mu(\cos\theta\ket{01}-e^{-i\gamma}\sin\theta\ket{10})\notag\\
\ket{e_2}&=-\sin\mu\ket{\Phi_3}+e^{i\nu}\cos\mu(\cos\theta\ket{01}-e^{-i\gamma}\sin\theta\ket{10}).\notag
\end{align}
Demanding that $tr_E(p_1\op{e_1}{e_1}+p_2\op{e_2}{e_2})=\mbb{I}/2$, we first arrive at the conditions
\begin{align}
\label{Eq:PurityExtract1}
\tfrac{1}{2}&=\tfrac{p\cos^2\mu+(1-p)\sin^2\mu}{2}+(p\sin^2\mu+(1-p)\cos^2\mu)\cos^2\theta\notag\\
&=\tfrac{p\cos^2\mu+(1-p)\sin^2\mu}{2}+(p\sin^2\mu+(1-p)\cos^2\mu)\sin^2\theta.
\end{align}
These can be simultaneously satisfied only if $0=(p\sin^2\mu+(1-p)\cos^2\mu)\cos 2\theta$, which has a solution if (a) $\cos^2\theta=\sin^2\theta=1/2$, or (b) $(p\sin^2\mu+(1-p)\cos^2\mu)=0$.  In case (a), Eq. \eqref{Eq:PurityExtract1} further requires that $p=1/2$.  In case (b), there is only one eigenvector $\ket{\Phi_3}$.  

In summary then, a unitary $U(\theta,\gamma)$ will $U$-generate the transformation $\mbb{I}/2\to\op{0}{0}$ iff  $\rho_{SE}=\op{\Phi_3}{\Phi_3}$ for arbitrary $(\theta,\gamma)$ unless  $\theta =\pi/4$.
However, if $\theta = \pi/4$, then $\rho_{SE}$ will be as following, 
\begin{equation}
\rho_{SE}=p\op{\Phi_3}{\Phi_3}+(1-p)\op{\Phi_4(\gamma)}{\Phi_4(\gamma)},
\end{equation}
where $\ket{\Phi_4(\gamma)}=\sqrt{1/2}(\ket{01}-e^{-i\gamma}\ket{10})$.  Note, this state is separable iff $p=1/2$; however, no product state solutions exist.

\begin{comment}

\subsection{Solutions in Higher Dimensions}

\textcolor{red}{Please add results and examples beyond two qubits...}

\medskip

\textcolor{red}{Ali, can you try to generalize Lemma 2 to higher dimensions.  Namely, try to prove that any purity extraction generated by a product state with $\rho_S$ being full rank requires $U\in\Swap$.  Or look for counterexamples....}

\end{comment}

\section{A Method for Detecting Initial Correlations}

\label{Sect:EntanglementDetection}

We now propose one way in which the questions studied above might be used for detecting initial correlations between a quantum system and an environment.  

\subsection{Certifying the Presence of System-Environment Entanglement}

Without additional knowledge of the system-environment dynamics, initial entanglement between the system and the environment cannot be decided by measuring the system alone.  However, if we know a certain transformation $\rho_S\to\rho_S'$ can only be $U$-generated by an entangled $\rho_{SE}$, then observing this transformation under the coupling $U$ guarantees that the system and environment are initially entangled.  Of course, here we are assuming that the experimenter has access to a source of preparations $\rho_{SE}$, and standard tomographic techniques are used to estimate $\rho_S$ and $\rho_S'$.  This is similar in spirit to the idea of detecting entanglement through the use of \textit{entanglement witnesses}, which involves an observable $W_{SE}$ for which $tr[W_{SE}\rho_{SE}]<0$ only if $\rho_{SE}$ is entangled \cite{Horodecki-1996a, Lewenstein-2000b}.  However, unlike entanglement witnesses, the detection scheme described here only requires measurements to be made on system $S$.  The ability to certify the presence of entanglement based on measurement data of $S$ alone comes from knowledge of the global unitary $U_{SE}$.

A natural and practically relevant question is whether system-environment entanglement can always be detected in this manner.  More precisely, given some $\rho_S$, can we always find a $U$-generated transformation $\rho_S\to\rho_S'$ that is possible only if $\rho_S$ is the reduced state of some entangled state $\rho_{SE}$?  The following theorem shows this is always possible.
\begin{theorem}
\label{Thm:Entangled}
For every genuinely mixed state $\rho_S$, there exists a unitary $U_{SE}$ such that the transformation $\rho_S\to\op{0}{0}_S$ can be $U$-generated only by an entangled $\rho_{SE}$.
\end{theorem}
\begin{proof}
Let $\rho_S$ have a spectral decomposition $\rho_S=\sum_{i=1}^rp_i\op{\Psi_i}{\Psi_i}$, where we assume $r\geq 2$.  For an arbitrary basis $\{\ket{j}_E\}_{j}$ for the environment, define a unitary $U_0$ whose action on basis states is
\begin{equation}
U_0\ket{\Psi_i}\ket{j}=\ket{i-j}\ket{j},
\end{equation}
where subtraction is taken mod $r$.  Then an arbitrary $\rho_{SE}$ for which $tr_E[\rho_{SE}]=\rho_S$ can be expanded as
\begin{equation}
\rho_{SE}=\sum_{i,j;i',j'} \alpha_{i,j;i',j'}\op{\Psi_i}{\Psi_{i'}}_S\otimes\op{j}{j'}_E.
\end{equation}
We compute
\begin{align}
U_0(\rho_{SE})U_0^\dagger&=\rho_{SE}'=\op{0}{0}\otimes\sigma_E\notag\\
&=\sum_{i,j;i',j'} \alpha_{i,j;i',j'}\op{i-j}{i'-j'}\otimes\op{j}{j'}.\notag
\end{align}
By considering the various contractions $_{E}\bra{j}\rho_{SE}'\ket{j'}_E$, it is easy to see that $i=j$ and $i'=j'$ for all the nonzero terms.  Hence the only compatible system-environment state is maximally correlated:
\begin{equation}
 \rho_{SE}=\sum_{ij}\alpha_{ij}\op{\Psi_i,i}{\Psi_j,j}
\end{equation} 
It is well-known that such states are always entangled unless $\rho_S$ is pure \cite{Rains-1999a}.  For completion, we supply a quick proof of this fact.  Taking a partial transpose of $\rho_{SE}$ gives 
\[\rho_{SE}^\Gamma=\sum_{i,j}\alpha_{ij}\op{\Psi_i,j}{\Psi_j,i}\]
From this we see that the support of $\rho_{SE}^\Gamma$ decomposes into one-dimensional subspaces (the $i^{th}$ spanned by $\ket{\Psi_i,i}$), and two-dimensional subspaces (the $i,j^{th}$ spanned by $\{\ket{\Psi_i,j},\ket{\Psi_j,i}\}$ for $i\not=j$).  Since $\rho$ is hermitian, on each of the two-dimensional subspaces, $\rho_{SE}^\Gamma$ has the form $\left(\begin{smallmatrix}0&\alpha_{ij}\\\alpha^*_{ij}&0\end{smallmatrix}\right)$.  Thus, the eigenvalues of $\rho_{SE}$ are $\alpha_{ii}$ and $\pm|\alpha_{ij}|$ for $i\not=j$.  By the PPT criterion of separability \cite{Peres-1996a}, this proves that $\rho_{SE}$ is entangled whenever $\rho_{SE}$ has rank greater than one.  
\end{proof}

\subsection{Certifying the Presence of System-Environment Correlations}

One largely practical drawback of Theorem \ref{Thm:Entangled} is that can be used to detect entanglement only if the system's evolved state is close to $\op{0}{0}$.  How small must $1-\bra{0}\rho_S'\ket{0}$ be in order to definitively certify entanglement?  We leave this seemingly complicated question open for future research.  However, if instead of initial entanglement, we focus on initial product state preparation, an experimentally useful answer can be given to this question. Recall that for a density matrix $\rho$, we let $\gamma(\rho)=tr[\rho^2]$. 
\begin{theorem}
\label{Thm:ProdState}
For every $\rho_S$, there exists a unitary $U_{SE}$ such that every transformation $\rho_S\to\rho_S'$ satisfying 
\[\bra{0}\rho_S'\ket{0}\leq \sqrt{\gamma(\rho_S)}\]
can only be $U$-generated by an initially correlated $\rho_{SE}$ (i.e. $\rho_{SE}$ cannot be a product state).
\end{theorem}
\begin{proof}
For $\rho_S=\sum_ip_i\op{\psi_i}{\psi_i}$, consider the unitary $U_{shift}$ defined in previous section.  If $\rho_E=\sum_{i',j'}c_{i'j'}\op{i'}{j'}$, then applying $U_{0}$ to an initial product state gives
\begin{align}
U_{shift}(\rho_S\otimes \rho_E )U_{shift}^\dagger=\sum_{i,i',j'}p_{i}c_{i'j'}\op{i-i',i'}{i-j',j'}.\notag
\end{align}
The final system state satisfies 
\begin{equation}
\bra{0}\rho_S'\ket{0}=\sum_{i}p_ic_{ii}\leq\sqrt{\sum_ip_i^2\sum_ic_{ii}^2}\leq\sqrt{\gamma(\rho_S)},\notag
\end{equation}
where we have used the Cauchy-Schwarz inequality.
\end{proof}
Theorem \ref{Thm:ProdState} provides an experimental criterion for detecting when $\rho_{SE}\not=\rho_S\otimes\rho_E$.  For many identical preparations of $\rho_{SE}$, the experimenter first uses tomographic techniques to estimate $\rho_S$ (actually the only knowledge of $\rho_S$ needed are its eigenstates and its purity).  Next, the unitary $U_{SE}$ is applied to the system and the interacting environment.  Finally, the experimenter measures $\rho_S'$ in the computational basis and estimates the value $\bra{0}\rho_S'\ket{0}$.  If the inequality in Theorem \ref{Thm:ProdState} is violated, the system and the environment must be initially correlated in each preparation of $\rho_{SE}$.

\section{Conclusion}

In this paper, we have begun investigating compatibility conditions between reduced-state dynamics $\rho_S\to\rho_S'$ and the underlying system-environment unitary evolution.  The motivating question has been how one can faithfully model such dynamics given partial knowledge of the system-environment interaction or of the initial system-environment joint state.  For example, it may be known that the interaction possesses certain symmetries which can be identified by eigenstates of the unitary $U_{SE}$.  In this case, we have shown in Lemma \ref{Lem:Ueigenstate} a necessary compatibility condition that must be satisfied which relates these eigenstates with the initial joint state $\rho_{SE}$ and the final system state $\rho'_{S}$.

If the unitary $U_{SE}$ is known completely, we considered whether its generated reduced-state dynamics can always be modeled in standard Stinespring form, using an initially uncorrelated system and environment state.  By examining two-qubit interactions, we proved in Theorem \ref{Thm:Main} that only product unitaries and SWAP always allow for such modeling.  Of course, reduced dynamics can always modeled by a standard Stinespring prescription if one places no restriction on $U_{SE}$ or considers the system and environment over a longer time interval.  However, such freedoms might not accurately reflect experimental situations, and our results confirm that a more general framework of open-system dynamics is needed in these cases.

As an application of our results, we describe in Theorems \ref{Thm:Entangled} and \ref{Thm:ProdState} how initial system-environment correlations can be witnessed by monitoring system evolution alone.  The potential to detect such correlations rests on prior knowledge of the system-environment interaction, but assumes no prior knowledge of the environment's state.  It would be interesting to see how this detection method can be strengthened.  For instance, the unitary described in Theorem \ref{Thm:ProdState} involves the ``shift'' operator $U_0$.  Perhaps statements like Theorem \ref{Thm:ProdState} can be made for more general classes of unitaries. 

A primary focus in this paper has been on transformations which we have called purity extractions.  The nature of these transformations make their quantitative analysis relatively simpler than more general transformations.  However, purity extractions are of fundamental interest from a thermodynamic perspective and when considering purity within a resource-theoretic framework \cite{Horodecki-2003b, Horodecki-2005a, delRio-2011a}.  We hope the work of this paper helps shed new light on these exciting topics as well as on general open system dynamics.

\begin{acknowledgments}
We would like to thank Kavan Modi for stimulating and constructive discussions on open-system dynamics.
\end{acknowledgments}

\bibliographystyle{apsrev4-1}

\bibliography{QuantumBib}

\section*{Appendix: Proof of Theorem 1}

\appendix

Before proceeding to the proof of Theorem 1, we provide two technical lemmas that explicate important properties of the classes $\Swap$ and $\CUU$.
\begin{proposition}
\label{Lem:LUCanonical}
Let $U_d$ be given as in Eq. \eqref{Eq:UnitaryCanonical}.  If $\lambda_i-\lambda_j\in\{0,\pi\}$ for all pairs $i$ and $j$, then $U_d$ is {\upshape LU} equivalent to either $\mbb{I}$ or $\mbb{F}$.
\end{proposition}
\begin{proof}
Note that $\LU$ and $\Swap$ are the only classes of two-qubit unitaries that map product states to product state.  So let $\ket{\tau}=\sum_{i=1}^4\alpha_i\ket{\Phi_i}$ be an arbitrary product state.  It is a fundamental property of the ``magic basis'' that no entanglement exists in $\ket{\tau}$ iff all the $\alpha_i$ are real (up to an overall phase).  Applying $U_d$ to $\ket{\tau}$ gives 
\[U_d\ket{\tau}=\sum_{i=1}^4e^{i\lambda_i}\alpha_i\ket{\Phi_i}=e^{i\lambda_j}\sum_{i=1}^4e^{i(\lambda_i-\lambda_j)}\alpha_i\ket{\Phi_i}\]
for any $j\in\{1,\cdots,4\}$.  This will be a product state only if $(\lambda_i-\lambda_j)\in\{0,\pi\}$.  Since $\ket{\tau}$ is an arbitrary product state, the proposition is proven.
\end{proof}
\begin{proposition}
\label{Lem:ProdBasis}
Suppose that $U$ performs the following pairwise transformation between orthogonal product states:
\begin{align}
\label{Eq:ProdBasis}
\ket{0}\ket{0}&\overset{\underset{U}{}}{\longrightarrow}\ket{0}\ket{0}&\ket{1}\ket{0}&\overset{\underset{U}{}}{\longrightarrow}\ket{1}\ket{b}\notag\\
\ket{a}\ket{1}&\overset{\underset{U}{}}{\longrightarrow}\ket{0}\ket{1}&\ket{a^\perp}\ket{1}&\overset{\underset{U}{}}{\longrightarrow}\ket{1}\ket{b^\perp}.
\end{align}
Then $U\in\CUU$.
\end{proposition}
\begin{proof}
The transformation is completed by first applying $\mathbb{I}\otimes\op{0}{0}+W\otimes\op{1}{1}$ and then $\op{0}{0}\otimes\mathbb{I}+\op{1}{1}\otimes V$, where $W$ (resp. $V$) rotates $\{\ket{a},\ket{a^\perp}\}$ (resp.  $\{\ket{b},\ket{b^\perp}\}$) into the computational basis.  Since a unitary is determined by its action on a complete basis, the conclusion of the lemma follows.
\end{proof}

\noindent\textit{Proof of Theorem \ref{Thm:Main}:}

By the decomposition of Eq. \eqref{Eq:UnitaryCanonical0}, it suffices to prove the theorem only for $U_d$-generated transformations.  Indeed, the transformation $\rho_S\to\rho_S'$ can be $U$-generated by some product/QC/separable state $\rho_{SE}$ iff the transformation $V_S \rho V_S^\dagger \to U_S^\dagger \sigma U_S$ is $U_d$-generated by the product/QC/separable state state $V_S\otimes V_E\rho_{SE} (V_S\otimes V_E)^\dagger$.

Let us first turn to the case when $U_d$ belongs $\Swap$.  By the same argument as just given, it suffices to consider when $U_d$ is $\mbb{F}$.  In this case, the transformation is $U_d$-generated by the product state $\rho_S\otimes\rho_S'$.  

Next, suppose that $U_d\in\CUU$.  Up to local unitaries, $U_d$ takes the form $U_d=U_c^{SE}U_c^{ES}$ where 
\begin{align}
\label{Eq:CUUcanonical}
U_c^{SE}&=\op{0}{0}\otimes\mbb{I}+\op{1}{1}\otimes U\notag\\
U_c^{ES}&=\mbb{I}\otimes\op{0}{0}+V\otimes \op{1}{1}
\end{align} 
with unitaries $U$ and $V$.    

Now, for any state $\rho_{SE}$, we can dephase in the computational basis: $\hat{\rho}_{SE}:=\sum_{i=1}^2\op{i}{i}_E\rho_{SE}\op{i}{i}_E$.  This is a QC state such that $tr_E\hat{\rho}_{SE}=\rho_S$, and direct calculation shows that it realizes the desired transformation under the action of $U_d$.  On the other hand, there does exist physical transformations $\rho_S\to\rho'_S$ that cannot be $U_d$-generated by any product state when $U_d\in\CUU\setminus\LU$.  First consider the case when $V\not=Z(\theta_1,\theta_2)$, where $Z(\theta_1,\theta_2)=\left(\begin{smallmatrix}e^{i\theta_1}&0\\0&e^{i\theta_2}\end{smallmatrix}\right)$ for arbitrary phases.  Then the state $\ket{\Psi}=\sqrt{1/2}(\ket{00}+(V^\dagger\otimes\mbb{I})\ket{01}$ is entangled.  Under the action of $U_d$ we have $\ket{\Psi}\to\ket{0}(U\ket{+})$, and hence the transformation on $S$ is $\rho_S=tr_E\op{\Psi}{\Psi}\to\op{0}{0}$ where $\rho_S$ is rank two.  If this transformation can be $U_d$-generated by a product state, then $U_d\rho_S\otimes\omega U_d^\dagger=\op{0}{0}\otimes\omega'$ for some $\omega$ and $\omega'$.  Since $\rho_S$ is rank two and $U_d$ preserves rank, we must have that $\omega=\op{\omega}{\omega}$ is a pure state and $\omega'$ is rank two.  Therefore, 
\[\rho_S\otimes \op{\omega}{\omega}=(U_c^{ES})^\dagger (U_c^{SE})^\dagger \op{0}{0}\otimes\omega'U_c^{SE}U_c^{ES}.\]
However, from Eq. \eqref{Eq:CUUcanonical}, $(U_c^{SE})^\dagger \op{0}{0}\otimes\omega'U_c^{SE}=\op{0}{0}\otimes\omega'$, and the final application of $U_c^{ES}$ will leave the environment invariant.  But this is impossible since $\omega'$ is rank two.  Hence, the transformation cannot be $U_d$-generated by a product state.  Now consider the case when $V\not=Z(\theta_1,\theta_2)$.  For such a unitary, it can be easily see that $U_d$ is LU equivalent to $U_c^{SE}=\op{0}{0}\otimes\mbb{I}=\op{1}{1}\otimes [UZ(0,\theta_2-\theta_1)]$.  Hence, it suffices to show a transformation $\rho_S\to\rho_S'$ that cannot be $U_c^{SE}$-generated by a product state.  This is done by repeating the same argument as just given except by choosing the initial entangled state $\ket{\Psi}=\sqrt{1/2}(\ket{00}+(\mbb{I}\otimes U^\dagger)\ket{10})$.  This state will indeed be entangled so long as $U\not=Z(\theta_1,\theta_2)$.  In the event that $U=Z(\theta_1,\theta_2)$ with $\theta_1\not=\theta_2$, simply use the initial entangled state $\sqrt{1/2}(\ket{0+}+(\mbb{I}\otimes U^\dagger)\ket{1+})$ in the previous argument.  When $U=Z(\theta_1,\theta_1)$, the operator $U_c^{SE}$ is LU.

We now consider unitaries $U_d$ not belonging to $\Swap\cup\CUU$.  By Lemma \ref{Lem:LUCanonical}, this means that $U_d$ will have at least one pair $\lambda_i$ and $\lambda_j$ for which $\lambda_i-\lambda_j\not\in\{0,\pi\}$.  We will show that for all such unitaries, there exists a physical transformation $\rho_S\to\rho_S'$ that cannot be $U_d$-generated by any separable state.

Without loss of generality, we can assume that $\lambda_3-\lambda_4\not\in\{0,\pi\}$.  Consequently, the states $\ket{\Psi^{\pm}}_{SE}=\sqrt{1/2}(e^{i\lambda_3}\ket{\Phi_3}\pm ie^{i\lambda_4}\ket{\Phi_4})$ are entangled, and so $\rho^\pm_S=tr_E\op{\Psi^{\pm}}{\Psi^{\pm}}$ are both full rank.  Applying $U_d$ to $\ket{\Psi^{\pm}}$ generates the state evolutions
\begin{align}
\ket{\Psi^\pm}\to U_d\ket{\Psi^\pm}=\sqrt{1/2}(\ket{\Phi_2}\pm\ket{\Phi_3})=\ket{\pm}\ket{\pm}.
\end{align}
Hence, the transformations $\rho^{\pm}_S\to\op{\pm}{\pm}$ are $U_d$-generated by the respective states $\op{\Psi^\pm}{\Psi^\pm}$.  Let us first just consider the transformation of $\rho^+_S$, and suppose this can be $U_d$-generated by some separable states $\tau$.  We then have $U_d\tau U_d^\dagger=\op{+}{+}\otimes\omega$ for some $\omega$,  and like before, $\tau$ must be rank two.  Any rank-two separable state of two qubits can be expressed as $\tau=\sum_{i=1}^2c_i\op{a_ib_i}{a_ib_i}$ \cite{Sanpera-1998a}, and each of the $\ket{a_ib_i}$ must transform into product states.  Thus, the following equalities hold true:
\begin{align}
\label{Eq:unitaries}
U_d\ket{a_1b_1}&=\ket{+}\ket{\beta_1}\notag\\
U_d\ket{a_2b_2}&=\ket{+}\ket{\beta_2}\notag\\
U_d\ket{\Psi^+}&=\ket{+}\ket{+},
\end{align}
where the $\ket{\beta_i}$ span the space of $E$.  We use them to express $\ket{+}=d_1\ket{\beta_1}+d_2\ket{\beta_2}$, and so with Eq. \eqref{Eq:unitaries}, we obtain $\ket{\Psi^+}=d_1\ket{a_1b_1}+d_2\ket{a_2b_2}$.  However, consistency of the initial state demands that both $\ket{\Psi^+}$ and $\ket{\tau}:=\sqrt{c_1}\ket{a_1b_1}_{SE}\ket{1}_{E'}+\sqrt{c_2}\ket{a_2b_2}_{SE}\ket{2}_{E'}$ are purifications of $\rho^{+}_S$.  Hence, there must exist an isometry $W:SE\to SEE'$ such that
\begin{align}
W\ket{\Psi^+}&=d_1\ket{a_1}\otimes W\ket{b_1}+d_2\ket{a_2}\otimes W\ket{b_2}&\notag\\
&=\sqrt{c_1}\ket{a_1}\ket{b_1}\ket{1}+\sqrt{c_2}\ket{a_2}\ket{b_2}\ket{2}.
\end{align}
Linear independence of the $\ket{a_i}$ means that $d_iW\ket{b_i}=\sqrt{c_i}\ket{a_i}\ket{b_i}\ket{i}$ and since the isometry preserves inner products, we must have that $\ip{b_1}{b_2}=0$.  From Eq. \eqref{Eq:unitaries}, this also implies that $\ip{\beta_1}{\beta_2}=0$.

We now repeat the same argument on the transformation of $\rho^-_S$.  Collectively we find that $U_d$ facilitates a transformation of product states taking the form
\begin{align}
U_d\ket{a_1b}&=\ket{+}\ket{\beta}&U_d\ket{\bar{a}_1\bar{b}}&=\ket{-}\ket{\bar{\beta}}\notag\\
U_d\ket{a_2b^\perp}&=\ket{+}\ket{\beta^\perp}&U_d\ket{\bar{a}_2\bar{b}^\perp}&=\ket{-}\ket{\bar{\beta}^\perp}.
\end{align}
Pairwise orthogonality further requires that either $\ket{\bar{a}_1\bar{b}}=\ket{a_1^\perp b}$ and $\ket{\bar{a}_2\bar{b}^\perp}=\ket{a_2^\perp b^\perp}$ or $\ket{\bar{a}_1\bar{b}}=\ket{a_2^\perp b^\perp}$ and $\ket{\bar{a}_2\bar{b}^\perp}=\ket{a^\perp_1b}$.  Either way, the transformation is of the form given in Eq. \eqref{Eq:ProdBasis}, and so by Lemma \ref{Lem:ProdBasis} we must have that $U_d\in\CUU$.  This is a contradiction.

\end{document}